\documentclass[10pt,english,reprint,superscriptaddress]{revtex4-1}
\usepackage[T1]{fontenc}
\usepackage[latin9]{inputenc}
\usepackage[letterpaper]{geometry}
\usepackage{amsmath}
\usepackage{amsthm}
\usepackage{amssymb}
\usepackage{graphicx}

\makeatletter
\theoremstyle{plain}
\newtheorem{thm}{\protect\theoremname}
  \theoremstyle{definition}
  \newtheorem{defn}[thm]{\protect\definitionname}
  \theoremstyle{plain}
  \newtheorem{lem}[thm]{\protect\lemmaname}
  \theoremstyle{remark}
  \newtheorem{rem}[thm]{\protect\remarkname}
  \theoremstyle{plain}
  \newtheorem{cor}[thm]{\protect\corollaryname}

\usepackage[all]{xy}

\usepackage{multirow}

\def\frontmatter@abstractheading{}

\newcommand{\se}{\mathsf{s}_{0}}
\newcommand{\so}{\mathsf{s}_{1}}

\usepackage{babel}
\providecommand{\definitionname}{Definition}
  \providecommand{\lemmaname}{Lemma}
\providecommand{\theoremname}{Theorem}

\usepackage{pdfpages}

\usepackage{babel}
\providecommand{\corollaryname}{Corollary}
  \providecommand{\definitionname}{Definition}
  \providecommand{\lemmaname}{Lemma}
  \providecommand{\remarkname}{Remark}
\providecommand{\theoremname}{Theorem}

\usepackage{babel}
\providecommand{\definitionname}{Definition}
\providecommand{\theoremname}{Theorem}

\usepackage{babel}
\providecommand{\definitionname}{Definition}
\providecommand{\theoremname}{Theorem}

\usepackage{babel}
\providecommand{\definitionname}{Definition}
\providecommand{\theoremname}{Theorem}

\makeatother

\usepackage{babel}
  \providecommand{\corollaryname}{Corollary}
  \providecommand{\definitionname}{Definition}
  \providecommand{\lemmaname}{Lemma}
  \providecommand{\remarkname}{Remark}
\providecommand{\theoremname}{Theorem}

\begin{document}

\title{Necessary and Sufficient Conditions for Extended Noncontextuality\\
 in a Broad Class of Quantum Mechanical Systems}

\author{Janne V.\ Kujala}

\email{To whom correspondence should be addressed. E-mail: jvk@iki.fi}

\affiliation{Department of Mathematical Information Technology, University of
Jyväskylä, Jyväskylä, Finland}

\author{Ehtibar N.\ Dzhafarov}

\email{E-mail: ehtibar@purdue.edu }

\affiliation{Department of Psychological Sciences, Purdue University, West Lafayette,
Indiana, USA}

\author{Jan-Åke Larsson}

\email{E-mail: jan-ake.larsson@liu.se}

\affiliation{Department of Electrical Engineering, Linköping University, 58183
Linköping, Sweden }
\begin{abstract}
The notion of (non)contextuality pertains to sets of properties measured
one subset (context) at a time.  We extend this notion to include
so-called inconsistently connected systems, in which the measurements
 of a given property in different contexts may have different distributions,
due to contextual biases in experimental design or physical interactions
(signaling): a system of measurements has a maximally noncontextual
description if they can be imposed a joint distribution on in which
the measurements of any one property in different contexts are equal
to each other with the maximal probability allowed by their different
distributions. We derive necessary and sufficient conditions for the
existence of such a description in a broad class of  systems including
Klyachko-Can-Binicio\u{g}lu-Shumvosky-type (KCBS), EPR-Bell-type,
and Leggett-Garg-type systems. Because these conditions allow for
inconsistent connectedness, they are applicable to real experiments.
We illustrate this by analyzing an experiment by Lapkiewicz and colleagues
aimed at testing contextuality in a KCBS-type system.


\markboth{Kujala, Dzhafarov, Larsson}{Criterion of Contextuality} 
\end{abstract}
\maketitle
The notion of (non)contextuality in Quantum Mechanics (QM) relates
the outcome of a measurement of a physical property $q$ to the choice
of properties $q',q'',\ldots$ co-measured with $q$ \cite{KochenSpecker1967}.
The set of co-measured properties $q,q',q'',\ldots$ forms a \emph{measurement
context} for each of its members. The traditional understanding of
a contextual QM system is that if the measurement of each property
$q$ in it is represented by a random variable $R_{q}$, then the
random variables representing all properties in the system do not
have a joint distribution.

We use here a different formulation, which, although formally equivalent,
lends itself to more productive development \cite{Larsson2002,DK2013PLOS,DK2014FooPh,DK2014PLOS,DK2015PhSc,DK2013LNCS}.
We label all measurements contextually: this means that a property
$q$ is represented by different random variables $R_{q}^{c}$ depending
on the context $c=\left\{ q,q',q'',\ldots\right\} $. We say that
the system has \emph{a noncontextual} description if there exists
a joint distribution of these random variables in which any two of
them, $R_{q}^{c_{1}}$ and $R_{q}^{c_{2}}$, representing the same
property $q$ in different contexts, are equal with probability 1.
If no such description exists we say that the system is \emph{contextual}.
Note that the existence of a joint distribution of several random
variables is equivalent to the possibility of presenting them as functions
of a single, ``hidden'' variable $\lambda$ \cite{SuppesZanotti1981,Fine1982,Larsson2002,DK2010}.

This formulation applies to systems in which the random variables
$R_{q}^{c_{1}},R_{q}^{c_{2}},\ldots$ representing a given property
in different contexts always have the same distribution. We call such
systems \emph{consistently connected}, because we call the set of
all such variables $R_{q}^{c_{1}},R_{q}^{c_{2}},\ldots$ for a given
$q$ a \emph{connection}. If the properties forming any given context
are space-time separated, consistent connectedness coincides with
the no-signaling condition \cite{PopescuRohrlich}. The central aim
of this paper is to extend the notion of contextuality to the cases
of \emph{inconsistent connectedness,} where the measurements of a
given property may have different distributions in different contexts.
This may happen due to a contextually biased measurement design or
due to physical influences exerted on $R_{q}^{c}$ by elements of
context $c$ other than $q$.

The criterion of (necessary and sufficient conditions for) contextuality
we derive below is formulated for inconsistently connected systems,
treating consistent connectedness as a special case. This makes it
applicable to real experimental data. For example, the experiment
in Ref.~\cite{Lapkiewicz2011} testing the Klyachko-Can-Binicio\u{g}lu-Shumvosky
(KCBS) inequality \cite{Klyachko2008} exhibits inconsistent connectedness,
necessitating a sophisticated work-around to establish contextuality
(see Refs.~\cite{Lapkiewicz2013,Ahrens2013}). Below, we apply our
extended notion to the same data to establish contextuality directly,
with no work-arounds. Another example is Leggett-Garg (LG) systems
\cite{LeggettGarg1985}, where our approach allows for the possibility
that later measurements may be affected by previous settings (``signaling
in time,'' \cite{Bacciagaluppi,Kofler2013}). Finally, in EPR-Bell-type
systems \cite{Bell1964,Bell1966} our approach allows for the possibility
that Alice's measurements are affected by Bob's settings \cite{Bacon-Toner 2003}
when they are time-like separated; and even with space-like separation,
the same effect can be caused by systematic errors \cite{Adenier-Khrennikov 2007}.

\emph{Earlier treatments.---} In the Kochen-Specker theorem \cite{KochenSpecker1967}
or its variants \cite{Peres1995,Cabello1996}, contexts are chosen
so that each property enters in more than one context, and in each
context, according to QM, one and only one of the measurements has
a nonzero value. The proof of contextuality, using our language, consists
in showing that the variables $R_{q}^{c}$ cannot be jointly assigned
values consistent with this constraint so that all the variables representing
the same property $q$ are assigned the same value. An experimental
test of contextuality here consists in simply showing that the observables
it specifies can be measured in the contexts it specifies, and that
the QM constraint in question is satisfied.

There has been recent work translating the value assignment proofs
into probabilistic inequalities (sometimes called \emph{Kochen-Specker
inequalities}) giving necessary conditions for noncontextuality \cite{Simon2001,Larsson2002}.
Inequalities that do not use value-assignment restrictions but only
the assumption of noncontextuality are known as \emph{noncontextuality
inequalities}~\cite{Cabello2008,Klyachko2008,Yu2012}. Bell inequalities
\cite{Bell1964,Bell1966,CHSH,CH1974,Fine1982} and LG inequalities
\cite{SuppesZanotti1981,LeggettGarg1985} are also established through
noncontextuality \cite{Mermin1993}, motivated by specific physical
considerations (locality and noninvasive measurement, resp.).

An extension of the notion of (non)contextuality that allows for inconsistent
connectedness was suggested in Refs.~\cite{Larsson2002,Winter2014}.
However, the error probability proposed in those papers as a measure
of context-dependent change in a random variable cannot be measured
experimentally. The suggestion in both Refs.~\cite{Larsson2002,Winter2014}
is to estimate the accuracy of the measurement and from that argue
for a particular value of the error probability. For example, Ref.~\cite{Winter2014}
uses the quantum description of the system for the estimate (quantum
tomography), but there is no clear reason why or how the quantum error
model would be related to that of the proposed noncontextual description.
A noncontextuality test should not mix the two descriptions, as it
attempts to show their fundamental differences.

In this paper we generalize the definition of contextuality in a different
manner, to allow for inconsistent connectedness while only using directly
measurable quantities. We derive a criterion of (non)contextuality
for a broad class of systems that includes as special cases the systems
intensively studied in the recent literature on contextuality: KCBS,
EPR-Bell, and LG systems \cite{Klyachko2008,Cabello2013,Cabello2014},
with their inconsistently connected versions \cite{DK_Bell_LG_K,DK_Bell_LG}.

\emph{Basic Concepts and Definitions.---} We begin by formalizing
the notation and terminology. Consider a finite set of distinct \emph{physical
properties $Q=\left\{ q_{1},\ldots,q_{n}\right\} $.} These properties
are measured in subsets of $Q$ called \emph{contexts}, $c_{1},\ldots,c_{m}$.
Let $C$ denote the set of all contexts, and $C_{q}$ the set of all
contexts containing a given property $q$.

The result of measuring property $q$ in context $c$ is a random
variable $R_{q}^{c}$ . The result of jointly measuring all properties
within a given context $c\in C$ is a set of jointly distributed random
variables $R^{c}=\left\{ R_{q}^{c}:q\in c\right\} $.

No two random variables in different contexts, $R_{q}^{c},R_{q'}^{c'}$,
$c\not=c'$, are jointly distributed, they are \emph{stochastically
unrelated} \cite{DK2013LNCS,DK2015PhSc}. The set of random variables
representing the same property $q$ in different contexts is called
a \emph{connection} (for $q$). So the elements of a connection $\left\{ R_{q}^{c}:c\in C_{q}\right\} $
are pairwise stochastically unrelated. If all random variables within
each connection are identically distributed, the system is called
\emph{consistently connected}; if it is not necessarily so, it is
\emph{inconsistently connected}. Consistent connectedness is also
known in QM as the \emph{Gleason property} \cite{CabelloSeverini2014},
outside physics as \emph{marginal selectivity} \cite{DK2013LNCS},
and Ref.~\cite{Cereceda2000} lists some dozen names for the same
notion; a recent addition to the list is \emph{no-disturbance principle}
\cite{Ramanathan2012,Kurzynski2014}.

The set $Q$ of all properties together with the set $C$ of all contexts
and the set $\left\{ R^{c}:c\in C\right\} $ of all sets of random
variables representing contexts is referred to as a \emph{system}.
In the systems we consider here the set of properties $q$ is finite
(whence the set of contexts $c$ is finite too), and each random variable
has a finite number of possible values (e.g., spin measurement outcomes).

We introduce next the notion of a (probabilistic) \emph{coupling}
of all the random variables $R_{q}^{c}$ in our system \cite{Thor}.
Intuitively, this is simply a joint distribution imposed, or ``forced''
on all of them (recall that they include stochastically unrelated
variables from different contexts). Formally, a coupling of $\left\{ R_{q}^{c}:q\in c\in C\right\} $
is any jointly distributed set of random variables $S=\left\{ S_{q}^{c}:q\in c\in C\right\} $
such that, for every $c\in C$, $\left\{ S_{q}^{c}:q\in c\right\} \sim\left\{ R_{q}^{c}:q\in c\right\} $,
where $\sim$ stands for ``has the same (joint) distribution as.''
One can also speak of a coupling for any subset of the random variables
$R_{q}^{c}$. Thus, fixing a property $q$, a coupling of a connection
$\left\{ R_{q}^{c}:c\in C_{q}\right\} $ is any jointly distributed
$\left\{ X_{q}^{c}:c\in C_{q}\right\} $ such that $X_{q}^{c}\sim R_{q}^{c}$
for all contexts $c\in C_{q}$. Note that if $S$ is a coupling of
all $R_{q}^{c}$, then every marginal (jointly distributed subset)
$\left\{ S_{q}^{c}:c\in C_{q}\right\} $ of $S$ is a coupling of
the corresponding connection $\left\{ R_{q}^{c}:c\in C_{q}\right\} $.

Expressed in this language, the traditional approach is to consider
a system\emph{ noncontextual} if there is a coupling $S$ of the random
variables $R_{q}^{c}$, such that for every property $q$ the random
variables in $\left\{ S_{q}^{c}:c\in C_{q}\right\} $ are equal to
each other with probability 1. That is, for every possible coupling
$S$ of the random variables $R_{q}^{c}$ and every property $q$
we consider the marginal $\left\{ S_{q}^{c}:c\in C_{q}\right\} $
corresponding to a connection $\left\{ R_{q}^{c}:c\in C_{q}\right\} $,
and we compute 
\begin{equation}
\Pr\left[S_{q}^{c_{q1}}=\cdots=S_{q}^{c_{qn_{q}}}\right],\;\left\{ c_{q1},\ldots,c_{qn_{q}}\right\} =C_{q}.\label{eq:equality in connection}
\end{equation}
If there exists a coupling $S$ for which this probability equals
1 for all $q$, this $S$ provides a noncontextual description for
our system. Otherwise, if in every possible coupling $S$ the probability
in question is less than 1 for some properties $q$, the system is
considered \emph{contextual}.

This understanding, however, only involves consistently connected
systems. As mentioned in the introduction, a system may be inconsistently
connected due to systematic biases or interactions (such as ``signaling
in time'' in LG systems). If for some $q$ and some contexts $c,c'\in C_{q}$,
the distribution of $R_{q}^{c}$ and $R_{q}^{c'}$ are not the same,
then $\Pr\big[S_{q}^{c}=S_{q}^{c'}\big]$ cannot equal 1 in any coupling
$S$. There would be nothing wrong if one chose to say that any such
inconsistently connected system is therefore contextual, but contextuality
due to systematic measurement errors or signaling is clearly a special,
trivial kind of contextuality. One should be interested in whether
the system exhibits any contextuality that is not reducible to (or
explainable by) the factors that make distributions of random variables
within a connection different. For systems in general therefore we
propose a different definition. 
\begin{defn}
\label{def:Main}A system has a \emph{maximally noncontextual description}
if there is a coupling $S$ of the random variables $R_{q}^{c}$,
such that for any $q$ the random variables $\left\{ S_{q}^{c}:c\in C_{q}\right\} $
in $S$ are equal to each other with the maximum probability allowed
by the individual distributions of $R_{q}^{c}$. 
\end{defn}
To explain, consider a connection $\left\{ R_{q}^{c}:c\in C_{q}\right\} $
in isolation, and let $\left\{ X_{q}^{c}:c\in C_{q}\right\} $ be
its coupling. Among all such couplings there must be \emph{maximal}
ones, those in which the probability that all variables in $\left\{ X_{q}^{c}:c\in C_{q}\right\} $
are equal to each other is maximal possible, given the distributions
of $X_{q}^{c}\sim R_{q}^{c}$. If a connection consists of two dichotomic
($\pm1$) variables $R_{q}^{1}$ and $R_{q}^{2}$, and $\left\{ X_{q}^{1},X_{q}^{2}\right\} $
is its coupling (i.e., $X_{q}^{1},X_{q}^{2}$ are jointly distributed
with $\left\langle X_{q}^{1}\right\rangle =\left\langle R_{q}^{1}\right\rangle $,
$\left\langle X_{q}^{2}\right\rangle =\left\langle R_{q}^{2}\right\rangle $),
then by Lemma A3 in Supplementary Material, the maximal possible expectation
$\left\langle X_{q}^{1}X_{q}^{2}\right\rangle $ is $1-\left|\left\langle R_{q}^{1}\right\rangle -\left\langle R_{q}^{2}\right\rangle \right|$;
a coupling $\left\{ X_{q}^{1},X_{q}^{2}\right\} $ with this expectation
is maximal. Now take every possible coupling $S$ of all our random
variables $R_{q}^{c}$, consider the marginals $\left\{ S_{q}^{c}:c\in C_{q}\right\} $
corresponding to connections $\left\{ R_{q}^{c}:c\in C_{q}\right\} $,
and for each of these marginals compute the probability (\ref{eq:equality in connection}).
If there is a coupling $S$ in which this probability equals its maximal
possible value for every $q$, this $S$ provides a maximally noncontextual
description for our system. For consistently connected systems Definition
1 reduces to the traditional understanding: the maximal probability
with which all variables in $\left\{ X_{q}^{c}:c\in C_{q}\right\} $
can be equal to each other is 1 if all these variables are identically
distributed.

\emph{Cyclic systems of dichotomic random variables.---} We focus
now on systems in which: (S1) each context consists of precisely two
distinct properties; (S2) each property belongs to precisely two distinct
contexts; and (S3) each random variable representing a property is
dichotomic ($\pm1$). As shown in Lemma A1 (Supplementary Material),
a set of properties satisfying S1--S2 can be arranged into one or
more distinct cycles $q_{1}\rightarrow q_{2}\rightarrow\ldots\rightarrow q_{k}\rightarrow q_{1}$,
in which any two successive properties form a context. Without loss
of generality we will assume that we deal with a \emph{single-cycle}
arrangement $q_{1}\rightarrow q_{2}\rightarrow\ldots\rightarrow q_{n}\rightarrow q_{1}$
of all the properties $\left\{ q_{1},\ldots,q_{n}\right\} $. The
number $n$ is referred to as the \emph{rank} of the system.

A schematic representation of a cyclic system is shown in Figure \ref{fig:n arbitrary}.
The LG paradigm exemplifies a cyclic system of rank $n=3$, on labeling
the observables $q_{1},q_{2},q_{3}$ measured chronologically. The
contexts $\left\{ q_{1},q_{2}\right\} ,\left\{ q_{2},q_{3}\right\} ,\left\{ q_{3},q_{1}\right\} $
here are represented by, respectively, pairs $\left(R_{1}^{1},R_{2}^{1}\right),\left(R_{2}^{2},R_{3}^{2}\right),\left(R_{3}^{3},R_{1}^{3}\right)$
with observed joint distributions, whereas $\left(R_{1}^{1},R_{1}^{3}\right),\left(R_{2}^{2},R_{2}^{1}\right),\left(R_{3}^{3},R_{3}^{2}\right)$
are connections for $q_{1},q_{2},q_{3}$, respectively. The EPR-Bell
paradigm exemplifies a cyclic system of rank $n=4$, on labeling the
observables $q_{1},q_{3}$ for Alice and $q_{2},q_{4}$ for Bob. Cyclic
systems of rank $n=5$ are exemplified by the KCBS paradigm, on labeling
the vertices of the KCBS pentagram by $q_{1}\rightarrow q_{2}\rightarrow q_{3}\rightarrow q_{4}\rightarrow q_{5}$.

\begin{figure}
\begin{centering}
\includegraphics[clip,scale=0.35]{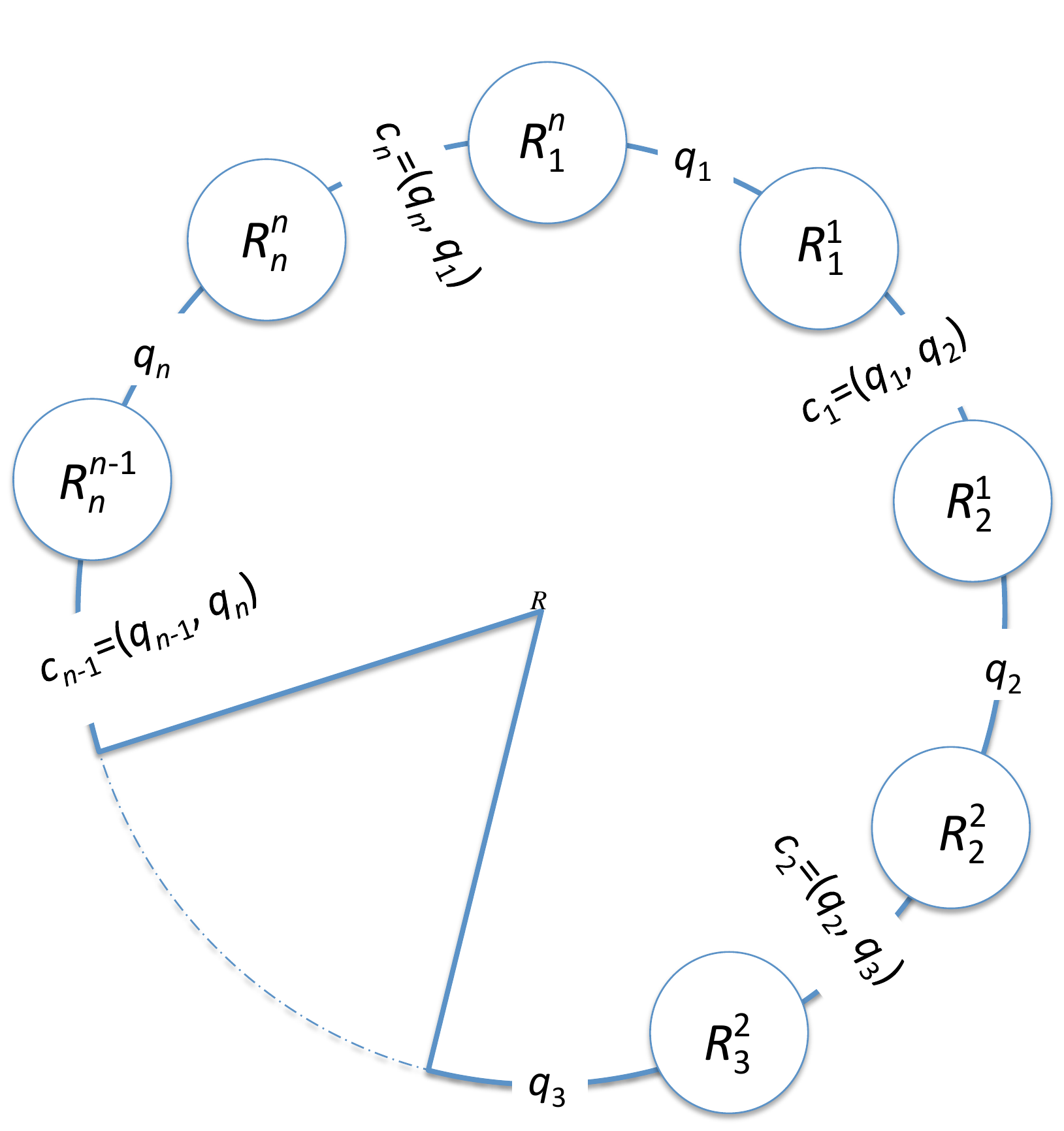} 
\par\end{centering}

\protect\caption{A schematic representation of a cyclic (single-cycle) system of rank
$n>1$. The properties $q_{1},\ldots,q_{n},q_{1}$ form a circle,
any two successive properties $\left(q_{i},q_{i\oplus1}\right)$ form
a context, denoted $c_{i}$ ($\oplus$ is clockwise shift $1\protect\mapsto2\protect\mapsto\ldots\protect\mapsto n\protect\mapsto1$).
In a given context $c_{i}$ the random variable representing $q_{i}$
is denoted $R_{i}^{i}$, and the one representing $q_{i\oplus1}$
is denoted $R_{i\oplus1}^{i}$. Each property $q_{i}$ therefore is
represented by two random variables: $R_{i}^{i}$ (when $q_{i}$ is
measured in context $c_{i}$) and $R_{i}^{i\ominus1}$ (when $q_{i}$
is measured in context $c_{i\ominus1}$). The pair $\left(R_{i}^{i\ominus1},R_{i}^{i}\right)$
is a connection for $q_{i}$, and the pair $\left(R_{i}^{i},R_{i\oplus1}^{i}\right)$
represents the context $c_{i}$.\label{fig:n arbitrary}}
\end{figure}

\emph{(Non)Contextuality Criterion.---} For any $n$, and any $x_{1},\dots,x_{n}\in\mathbb{R}$,
we define the function 
\begin{equation}
\so\left(x_{1},\ldots,x_{n}\right)=\max_{\iota_{1},\ldots,\iota_{n}\in\left\{ -1,1\right\} ,\prod_{k}\iota_{k}=-1}\sum_{k}\iota_{k}x_{k}.\label{eq:function so}
\end{equation}
The maximum is taken over all combinations of $\pm1$ coefficients
$\iota_{1},\ldots,\iota_{n}$ containing odd numbers of $-1$'s. The
following is our main theorem. 
\begin{thm}
\label{thm:Main Criterion}A cyclic system of rank $n>1$ with dichotomic
random variables (see Figure \ref{fig:n arbitrary}) has a maximally
noncontextual description if and only if 
\begin{equation}
\so\left(\left\langle R_{i}^{i}R_{i\oplus1}^{i}\right\rangle ,1-\left|\left\langle R_{i}^{i}\right\rangle -\left\langle R_{i}^{i\ominus1}\right\rangle \right|:i=1,\ldots,n\right)\le2n-2\label{eq:master}
\end{equation}
($\so$ here having $2n$ arguments, each entry being taken with $i=1,\ldots,n$
). 
\end{thm}
See Supplementary Material for the proof. In \eqref{eq:master}, $\left\langle R_{i}^{i}R_{i\oplus1}^{i}\right\rangle $
are the quantum correlations observed within contexts, whereas $1-\left|\left\langle R_{i}^{i}\right\rangle -\left\langle R_{i}^{i\ominus1}\right\rangle \right|$
are the maximal values for the unobservable correlations within the
couplings of connections. If the system is consistently connected,
i.e., $\left\langle R_{i}^{i}\right\rangle =\left\langle R_{i}^{i\ominus1}\right\rangle $,
then these maximal values equal 1. By Corollary A10, the criterion
\eqref{eq:master} then reduces to the formula 
\begin{equation}
\so\left(\left\langle R_{i}^{i}R_{i\oplus1}^{i}\right\rangle :i=1,\ldots,n\right)\le n-2,\label{eq:CC}
\end{equation}
well-known for $n=3$ (the LG inequality in the form derived in Ref.~\cite{SuppesZanotti1981})
and for $n=4$ (CHSH inequalities \cite{CHSH}). For $n=5$, \eqref{eq:CC}
contains the KCBS inequality (which by Corollary A.11 is not only
necessary but also sufficient for the existence of a maximally noncontextual
description). Finally, for any even $n\geq4$, inequality \eqref{eq:CC}
contains the \emph{chained Bell inequalities} studied in Refs.~\cite{Pearle1970,BraunstenCaves1990}.
It is known that for $n>4$ the chained Bell inequalities are not
criteria, the latter requiring many more inequalities \cite{WW2001a,WW2001b,BasoaltoPercival2003,DK2012}.

Generally, some of the terms $\left\langle R_{i}^{i}\right\rangle -\left\langle R_{i}^{i\ominus1}\right\rangle $
in \eqref{eq:master} may be nonzero. Thus, in an LG system ($n=3$),
if inconsistency is due to ``signaling in time'' \cite{Bacciagaluppi,Kofler2013},
these may include $\left\langle R_{2}^{2}\right\rangle -\left\langle R_{2}^{1}\right\rangle $
and $\left\langle R_{3}^{3}\right\rangle -\left\langle R_{3}^{2}\right\rangle $
but not $\left\langle R_{1}^{1}\right\rangle -\left\langle R_{1}^{3}\right\rangle $,
because $q_{1}$ cannot be influenced by later events. However, $\left\langle R_{1}^{1}\right\rangle -\left\langle R_{1}^{3}\right\rangle $
may be nonzero due to contextual biases in design, if something in
the procedure of measuring $q_{1}$ is different depending on whether
the next measurement is going to be of $q_{2}$ or $q_{3}$.

\emph{An application to experimental data.---} To illustrate the applicability
of our theory to real experiments, consider the data from the KCBS
experiment of Ref.~\cite{Lapkiewicz2011}. The experiment uses a
single photon in a quantum overlap of three optical modes (paths)
as an indivisible quantum system. Readout is performed through single-photon
detectors that terminate the three paths. Context is chosen through
``activation\textquotedblright{} of transformations, by rotating
a wave-plate that precedes each beamsplitter to change the behavior
of two out of three paths. Each transformation leaves one path untouched,
which serves as justification for consistent connectedness of the
corresponding measurements, $\langle R_{i}^{i}\rangle=\langle R_{i}^{i\ominus1}\rangle$,
so that the target inequality is (\ref{eq:CC}) for $n=5$.

$R_{1}^{1}$ and $R_{1}^{5}$ are recorded in different experimental
setups with zero or four polarizing beamsplitters ``activated''.
These outputs have significantly different distributions: from Ref.~\cite{Lapkiewicz2011}
Table~1, $\langle R_{1}^{1}\rangle=.136(6)$, $\langle R_{1}^{5}\rangle=.172(4)$,
and taking them as means and standard errors of 20 replications, the
standard $t$-test with $\textnormal{df}=19$ is significant at 0.1\%.
Lapkiewicz et al., deal with this by introducing in (\ref{eq:CC})
a correction term involving $\langle R_{1}^{1}R_{1}^{5}\rangle$.
They estimate $\langle R_{1}^{1}R_{1}^{5}\rangle$ by identifying
$R_{1}^{1}$ with $R'_{1}$, an output measured in a separate context
and in a special manner: instead of photon detections it is measured
by blocking two paths early in the setup. While this results in a
well-motivated experimental test, the identification of $R'_{1}$
with $R_{1}^{1}$ involves additional assumptions \cite{Ahrens2013,Lapkiewicz2013}.
Furthermore, Lapkiewicz et al. have to discount the fact that the
assumption $\langle R_{i}^{i}\rangle=\langle R_{i}^{i\ominus1}\rangle$
can also be challenged for $i=4$: the same $t$-test as above for
$\langle R_{4}^{4}\rangle=.122(4)$ and $\langle R_{4}^{3}\rangle=.142(4)$
is significant at 1\%. We see that the traditional approach adopted
in Ref.~\cite{Lapkiewicz2011} encounters considerable experimental
and analytic difficulties due to the necessity of avoiding inconsistent
connectedness.

Our theory allows one to analyze the data directly as found in the
measurement record. It is convenient to do this by using the inequality
\begin{equation}
\so\big(\left\langle R_{i}^{i}R_{i\oplus1}^{i}\right\rangle :i=1,\ldots,n\big)-\sum_{i=1}^{n}\left|\left\langle R_{i}^{i}\right\rangle -\left\langle R_{i}^{i\ominus1}\right\rangle \right|\le n-2,\label{eq:simpleform}
\end{equation}
which, by Corollary A9, follows from the criterion \eqref{eq:master}
\cite{remark_conjecture}. One way of using it is to construct a conservative
$100\left(1-\alpha\right)\%$ confidence interval with, say, $\alpha=10^{-10}$
for the left-hand side of \eqref{eq:simpleform} with $n=5$ and show
that its lower endpoint exceeds $n-2=3$. One can, e.g., construct
10 Bonferroni $100\left(1-{\alpha}/{10}\right)\%$ confidence intervals
for each of the approximately normally distributed terms $\left\langle R_{i}^{i}R_{i\oplus1}^{i}\right\rangle $
and $\left\langle R_{i}^{i}\right\rangle -\left\langle R_{i}^{i\ominus1}\right\rangle $
($i=1,\ldots,5$), with respective error terms read or computed from
Table 1 of Ref.~\cite{Lapkiewicz2011}, and then determine the range
of \eqref{eq:simpleform}. Treating each estimated term as the mean
of 20 observations, we have $t_{1-\alpha/10}\left(19\right)<14$ and
so a conservative confidence interval for each term is given by $\pm14\times\text{standard error}$.
Using these intervals, we can calculate the conservative $100\left(1-10^{-10}\right)\%$
confidence interval for \eqref{eq:simpleform} as 
\begin{widetext}
\begin{equation}
\begin{split} & \so\Big(\overbrace{\left\langle R_{1}^{1}R_{2}^{1}\right\rangle }^{-.805\pm.028},\overbrace{\left\langle R_{2}^{2}R_{3}^{2}\right\rangle }^{-.804\pm.042},\overbrace{\left\langle R_{3}^{3}R_{4}^{3}\right\rangle }^{-.709\pm.042},\overbrace{\left\langle R_{4}^{4}R_{5}^{4}\right\rangle }^{-.810\pm.028},\overbrace{\left\langle R_{5}^{5}R_{1}^{5}\right\rangle }^{-.766\pm.028}\Big)\\
 & \quad-\big|\underbrace{\left\langle R_{1}^{1}\right\rangle -\left\langle R_{1}^{5}\right\rangle }_{-.036\pm.101}\big|-\big|\underbrace{\left\langle R_{2}^{2}\right\rangle -\left\langle R_{2}^{1}\right\rangle }_{-.004\pm.140}\big|-\big|\underbrace{\left\langle R_{3}^{3}\right\rangle -\left\langle R_{3}^{2}\right\rangle }_{.006\pm.126}\big|-\big|\underbrace{\left\langle R_{4}^{4}\right\rangle -\left\langle R_{4}^{3}\right\rangle }_{-.020\pm.080}\big|-\big|\underbrace{\left\langle R_{5}^{5}\right\rangle -\left\langle R_{5}^{4}\right\rangle }_{-.006\pm.080}\big|=[3.127,4.062].
\end{split}
\label{eq:1}
\end{equation}

\end{widetext}

The system is contextual. The conclusion is the same as in Ref.~\cite{Lapkiewicz2011},
but we arrive at it by a shorter and more robust route.

\emph{Conclusion.---} We have derived a criterion of (non)contextuality
applicable to cyclic systems of arbitrary ranks. Even for consistently
connected systems this criterion has not been previously known for
ranks $n\geq5$ (KCBS and higher-rank systems). However, it is the
inclusion of inconsistently connected systems that is of special interest,
because it makes the theory applicable to real experiments. A ``system''
is not just a system of properties being measured, but also a system
of measurement procedures being used, with possible contextual biases
and unaccounted-for interactions. Our analysis opens the possibility
of studying contextuality without attempting to eliminate these first,
whether by statistical analysis or by improved experimental procedure. 
\begin{acknowledgments}
This work is supported by NSF grant SES-1155956, AFOSR grant FA9550-14-1-0318,
A. von Humboldt Foundation, and FQXi through Silicon Valley Community
Foundation. We thank J. Acacio de Barros, Gary Oas, Samson Abramsky,
Guido Bacciagaluppi, Adán Cabello, Andrei Khrennikov, and Lasse Leskelä
for numerous discussions. \end{acknowledgments}

\clearpage{}\newpage{}

\global\long\def\thesection{A}

\section*{Supplementary Material to\protect \protect \\
 ``Necessary and Sufficient Conditions for Maximal Noncontextuality
in a Broad Class of Quantum Mechanical Systems.'' Proof of the main
criterion and its consequences}

\numberwithin{equation}{section}\numberwithin{thm}{section}\numberwithin{figure}{section}

\setcounter{equation}{0}\setcounter{thm}{0}

The (non)contextuality criterion derived in this main text is a corollary
to Theorem \ref{thm:A1-An} proved below. We first need the following
simple result (see properties S1 and S2 formulated in section \emph{Cyclic
systems of dichotomic random variables}): 
\begin{lem}
In a system satisfying S1-S2, the physical properties $\left\{ q_{1},\ldots,q\right\} $
can be (re)indexed and arranged in one or more non-overlapping cycles
\begin{equation}
\left(q_{11},\ldots,q_{1n_{1}},q_{11}\right),\left(q_{21},\ldots,q_{2n_{2}},q_{21}\right),\ldots,\left(q_{k1},\ldots,q_{kn_{k}},q_{k1}\right),
\end{equation}
with $n_{1}+\ldots+n_{k}=n$ and $n_{i}>2$ ($i=1,\ldots,k$), such
that any two successive properties in each cycle form a context. \end{lem}
\begin{proof}
Apparent from Figure \ref{fig:arrangements}. 
\end{proof}
Our proof of Theorem \ref{thm:A1-An} uses the fact that the connections
and context representations enter a circular system symmetrically,
so that it is possible to view circular systems as a circular arrangement
of random variables $A_{1},\ldots,A_{n},A_{1}$ in which any two successive
variables have a joint distribution (see Figure \ref{fig:A1-An-1}).

We need some auxiliary results. In addition to $\so$ defined in the
main text, we use function

\begin{equation}
\se\left(x_{1},\ldots,x_{n}\right)=\max_{\iota_{1},\ldots,\iota_{n}\in\left\{ -1,1\right\} ,\prod_{k}\iota_{k}=1}\sum_{k}\iota_{k}x_{k},
\end{equation}
in which the maximum is taken over all combinations of $\pm1$ coefficients
$\iota_{1},\ldots,\iota_{n}$ containing even numbers of $-1$'s. 
\begin{lem}
\label{lem:s0_s1_split}For any $a_{1},\dots,a_{n},b_{1},\dots,b_{m}\in\mathbb{R}$,
\begin{equation}
\begin{array}{l}
\so(a_{1},\dots,a_{n},b_{1},\dots,b_{m})\\
=\max\begin{array}[t]{c}
\left\{ \begin{array}{c}
\se(a_{1},\dots,a_{n})+\so(b_{1},\dots,b_{m}),\\
\so(a_{1},\dots,a_{n})+\se(b_{1},\dots,b_{m})
\end{array}\right\} ,\end{array}
\end{array}
\end{equation}
and 
\end{lem}
\begin{equation}
\begin{array}{l}
\se(a_{1},\dots,a_{n},b_{1},\dots,b_{m})\\
=\max\begin{array}[t]{c}
\left\{ \begin{array}{c}
\se(a_{1},\dots,a_{n})+\se(b_{1},\dots,b_{m}),\\
\so(a_{1},\dots,a_{n})+\so(b_{1},\dots,b_{m})
\end{array}\right\} .\end{array}
\end{array}
\end{equation}
The proof is obvious. 
\begin{lem}
\label{lem:joint-of-two}Jointly distributed $\pm1$-valued random
variables $A$ and $B$ with given expectations $\left\langle A\right\rangle ,\left\langle B\right\rangle ,\left\langle AB\right\rangle $
exist if and only if 
\begin{equation}
\begin{array}{c}
-1\le\left\langle A\right\rangle \le1,\\
-1\le\left\langle B\right\rangle \le1,\\
\left|\left\langle A\right\rangle +\left\langle B\right\rangle \right|-1\le\left\langle AB\right\rangle \le1-\left|\left\langle A\right\rangle -\left\langle B\right\rangle \right|.
\end{array}\label{eq:Lemma 2}
\end{equation}
\end{lem}
\begin{proof}
For jointly distributed $\left(A,B\right)$, from the table of probabilities

\[
\begin{tabular}{|c|cc|cccccccccc}
 \cline{2-3} \multicolumn{1}{c|}{}  &  \ensuremath{B=+1}  &  \ensuremath{B=-1}  &  \tabularnewline\cline{1-3} \ensuremath{A=+1}  &  \ensuremath{r}  &  \ensuremath{p-r}  &  \ensuremath{p}\tabularnewline\ensuremath{A=-1}  &  \ensuremath{q-r}  &  \ensuremath{1-p-q+r}  &  \ensuremath{1-p}\tabularnewline\cline{1-3} \multicolumn{1}{c}{}  &  \ensuremath{q}  &  \multicolumn{1}{c}{\ensuremath{1-q}}  &  \tabularnewline\end{tabular}
\]

\noindent it is clear that 
\begin{equation}
\max\left(p+q-1,0\right)\leq r\leq\min\left(p,q\right).\label{eq:Lemma 2 equivalent}
\end{equation}
Since 
\[
\begin{array}{c}
\left\langle AB\right\rangle =1-2p-2q+4r,\\
\left\langle A\right\rangle =2p-1,\\
\left\langle B\right\rangle =2q-1,
\end{array}
\]
straightforward algebra leads to (\ref{eq:Lemma 2}). Conversely,
expressing $p,q,r$ through $\left\langle A\right\rangle ,\left\langle B\right\rangle ,\left\langle AB\right\rangle $,
(\ref{eq:Lemma 2}) implies $\left(\ref{eq:Lemma 2 equivalent}\right)$,
and then all probabilities in the table above are well defined. 
\end{proof}
\begin{figure}
\begin{centering}
\includegraphics[scale=0.25]{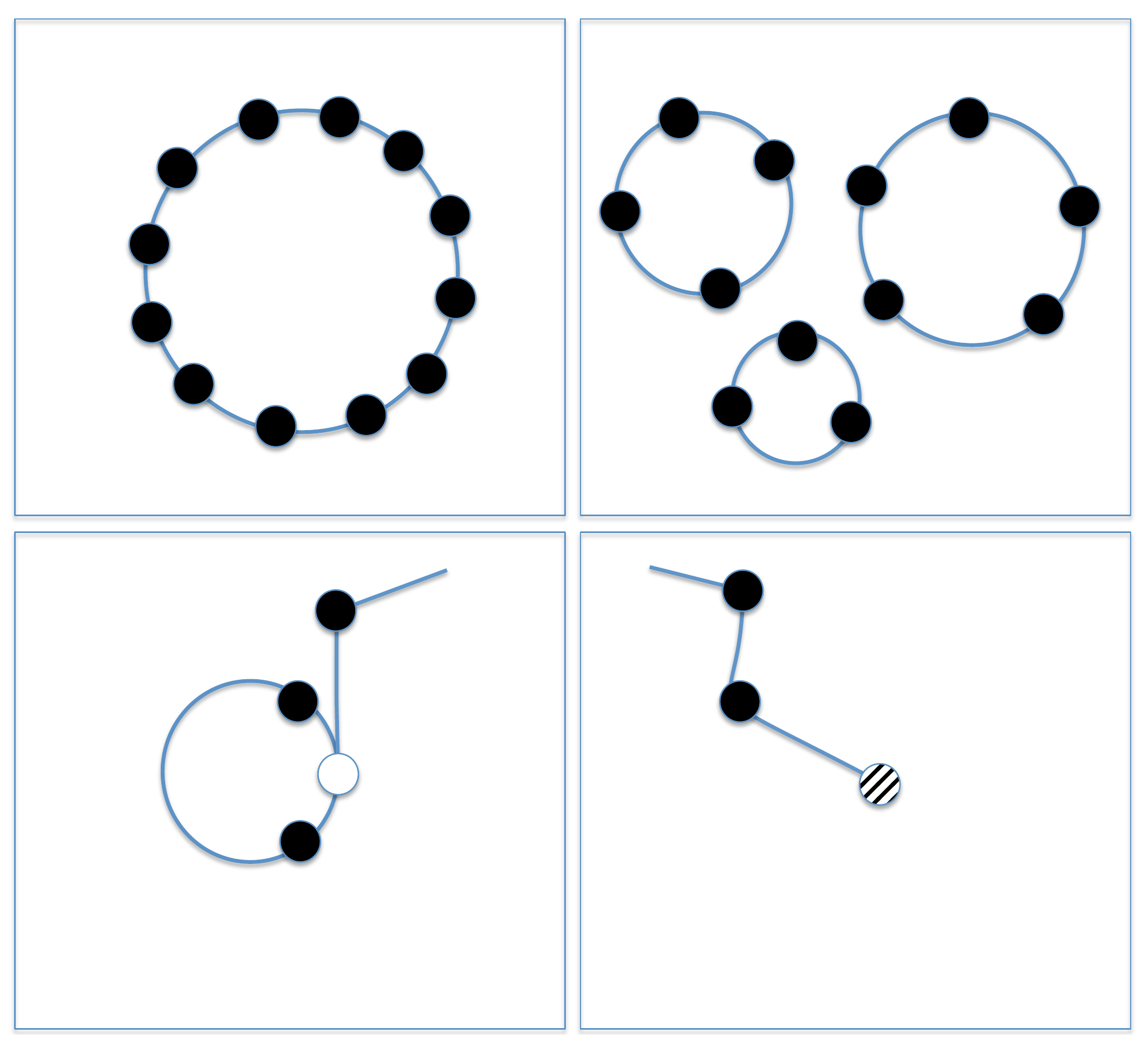} 
\par\end{centering}

\protect\protect\protect\protect\caption{{\footnotesize{}{}{}In a system satisfying S1-S2 the properties
being measured (represented by small circles) can be arranged in one
(top left) or more (top right) cycles in which any two successive
elements form a context. The bottom panels show that no other arrangements
are possible: the patterned circle participates in less than two contexts,
the open circle belongs to more than two contexts.}\label{fig:arrangements}}
\end{figure}

\begin{lem}
\label{lem:joint-of-three}Jointly distributed $\pm1$-valued random
variables $A$, $B$, and $C$ with given expectations $\left\langle A\right\rangle $,
$\left\langle B\right\rangle $, $\left\langle C\right\rangle $ $\left\langle AB\right\rangle $,
$\left\langle AC\right\rangle $, $\left\langle BC\right\rangle $
exist if and only if these expectations satisfy Lemma \ref{lem:joint-of-two}
and 
\begin{equation}
\so\left(\left\langle AB\right\rangle ,\left\langle BC\right\rangle ,\left\langle CA\right\rangle \right)\le1.
\end{equation}
\end{lem}
\begin{proof}
$\left\langle AB\right\rangle $, $\left\langle A\right\rangle $,
and $\left\langle B\right\rangle $ satisfying Lemma \ref{lem:joint-of-two}
uniquely determine $\Pr\left[A=1,B=1\right]$; and analogously for
$\Pr\left[B=1,C=1\right]$ and $\Pr\left[C=1,A=1\right]$. A joint
distribution of $\left(A,B,C\right)$ is determined by 8 probabilities
$p_{abc}=\Pr\left[A=a,B=b,C=c\right]$, $a,b,c\in\left\{ -1,1\right\} $.
It has the given expectations if and only if the $8$ probabilities
$p_{abc}$ satisfy 7 equations 
\[
\begin{array}{cc}
\sum_{b,c}p_{1bc}=\Pr\left[A'=1\right], & \sum_{c}p_{11c}=\Pr\left[A'=1,B'=1\right],\\
\\
\sum_{a,c}p_{a1c}=\Pr\left[B'=1\right], & \sum_{a}p_{a11}=\Pr\left[B'=1,C'=1\right],\\
\\
\sum_{a,b}p_{ab1}=\Pr\left[C'=1\right], & \sum_{b}p_{1b1}=\Pr\left[C'=1,A'=1\right],
\end{array}
\]
\[
\sum_{a,b,c}p_{abc}=1.
\]
The statement of the lemma obtains by any algorithm (facet enumeration
and reduction) analogous to that described in Text~S3 of Ref. \cite{DK2013PLOS-1}.\end{proof}
\begin{rem}
One can also obtain the proof by using Fine's theorem \cite{Fine1982-1},
presenting it as (using Fine's notation for the random variables)
\[
\so\left(\left\langle A_{1}B_{1}\right\rangle ,\left\langle A_{1}B_{2}\right\rangle ,\left\langle A_{2}B_{1}\right\rangle ,\left\langle A_{2}B_{2}\right\rangle \right)\le2,
\]
and then putting $A_{1}=B_{1}=A$, $B_{2}=B$, and $A_{2}=C$.\end{rem}
\begin{lem}
\label{lem:chain-joint}Jointly distributed arbitrary random variables
$A,B,C$ with given 2-marginal distributions of $(A,B)$ and $(B,C)$
exist if and only if these 2-marginals agree for the distribution
of $B$.\end{lem}
\begin{proof}
The necessity is obvious. The sufficiency obtains by the Markov rule
\[
\begin{array}{l}
\Pr[A=a,B=b,C=c]\\
=\Pr[C=c\mid B=b]\Pr[B=b\mid A=a]\Pr[A=a],
\end{array}
\]
for any possible values $a,b,c$ of, respectively, $A,B,C$. \end{proof}
\begin{cor}[to Lemma \ref{lem:chain-joint}]
\label{cor:chain-joint}Jointly distributed $\pm1$-valued random
variables $A_{1},\dots,A_{n}$ with given expectations $\left\langle A{}_{1}\right\rangle ,\ldots,\left\langle A_{n}\right\rangle $,
$\left\langle A{}_{1}A{}_{2}\right\rangle ,\dots,\left\langle A{}_{n-1}A{}_{n}\right\rangle $
exist if and only if these expectations satisfy Lemma \ref{lem:joint-of-two}. 
\end{cor}
\begin{figure}
\begin{centering}
\includegraphics[scale=0.35]{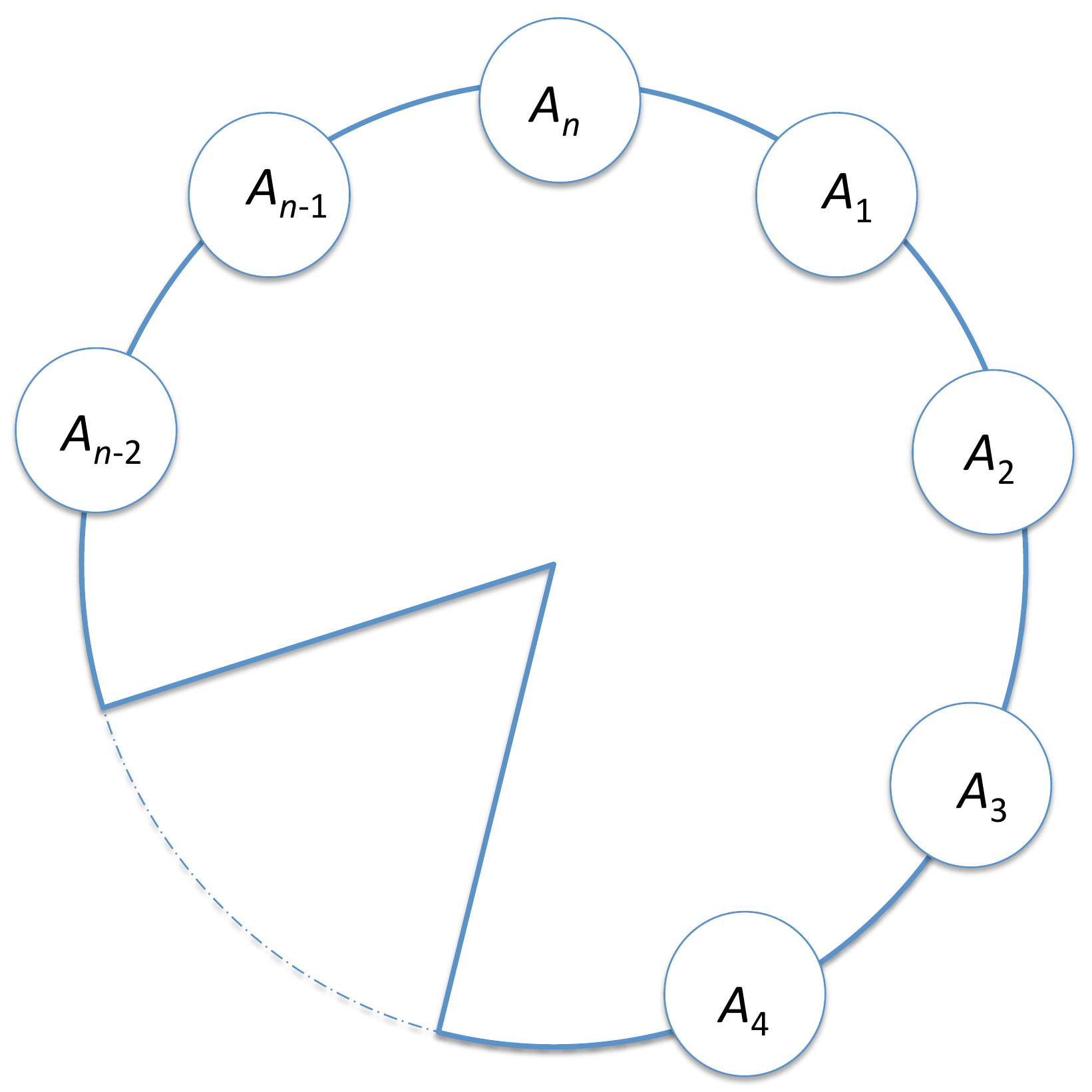} 
\par\end{centering}

\protect\protect\caption{{\footnotesize{}{}The arrangement of random variables with no distinction
made between context representations and connections: $n$ random
variables any successive two of which have a joint distribution. This
structure is characterized in Theorem \ref{thm:A1-An}. \label{fig:A1-An-1}}}
\end{figure}

\begin{thm}
\label{thm:A1-An}Jointly distributed $\pm1$-valued random variables
$A_{1},\dots,A_{n}$ ($n\ge3$) with given expectations 
\begin{equation}
\left\langle A{}_{1}A{}_{2}\right\rangle ,\dots,\left\langle A{}_{n-1}A{}_{n}\right\rangle ,\left\langle A{}_{n}A{}_{1}\right\rangle \label{eq:expectations_n}
\end{equation}
exist if and only if these expectations satisfy Lemma \ref{lem:joint-of-two}
and 
\begin{equation}
\so\left(\left\langle A{}_{1}A{}_{2}\right\rangle ,\dots,\left\langle A{}_{n-1}A{}_{n}\right\rangle ,\left\langle A{}_{n}A{}_{1}\right\rangle \right)\le n-2.\label{eq:main lemma}
\end{equation}
\end{thm}
\begin{proof}
For $n=3$ the statement follows from Lemma \ref{lem:joint-of-three}.

Assume that the statement holds up to and including some $n\ge3$.
We will prove that

(i) jointly distributed $\pm1$-valued random variables $A_{1},\dots,A_{n+1}$
with given expectations 
\[
\left\langle A{}_{1}\right\rangle ,\ldots,\left\langle A{}_{n+1}\right\rangle ,\left\langle A{}_{1}A{}_{2}\right\rangle ,\dots,\left\langle A{}_{n}A{}_{n+1}\right\rangle ,\left\langle A{}_{n+1}A{}_{1}\right\rangle 
\]
exist if and only if

(ii) these expectations satisfy Lemma \ref{lem:joint-of-two} and

(iii) they satisfy

\[
\so\left(\left\langle A{}_{1}A{}_{2}\right\rangle ,\dots,\left\langle A{}_{n}A{}_{n+1}\right\rangle ,\left\langle A{}_{n+1}A{}_{1}\right\rangle \right)\le n-1.
\]
Since Statement (ii) is an obvious consequence of Statement (i), we
only need to prove that if Statement (ii) is satisfied, then Statements
(i) and (iii) are equivalent. So we assume Statement (ii).

By Lemma \ref{lem:chain-joint}, jointly distributed $\left(A_{2},\ldots,A_{n-1}\right),\left(A_{1},A_{n}\right),A_{n+1}$
exist if and only if there are jointly distributed $\left(A_{2},\ldots,A_{n-1}\right),\left(A_{1},A_{n}\right)$
and $\left(A_{1},A_{n}\right),A_{n+1}$, with one and the same jointly
distributed $\left(A_{1},A_{n}\right)$. Hence Statement (i) holds
if and only if, for some $\left\langle A_{n}A_{1}\right\rangle $
satisfying Lemma \ref{lem:joint-of-two}, $\left(A{}_{1},\dots,A{}_{n}\right)$
exists with expectations $\left\langle A{}_{1}\right\rangle ,\ldots,\left\langle A{}_{n}\right\rangle ,\left\langle A{}_{1}A{}_{2}\right\rangle ,\dots,\left\langle A{}_{n-1}A{}_{n}\right\rangle ,\left\langle A{}_{n}A{}_{1}\right\rangle $,
and $\left(A{}_{n},A{}_{n+1},A{}_{1}\right)$ exists with expectations
$\left\langle A{}_{1}\right\rangle ,\left\langle A{}_{n}\right\rangle ,\left\langle A{}_{n+1}\right\rangle ,\left\langle A{}_{1}A{}_{n}\right\rangle ,\left\langle A{}_{n}A{}_{n+1}\right\rangle ,\left\langle A{}_{n+1}A{}_{1}\right\rangle $.
Therefore, by the induction hypothesis, Statement (i) holds if and
only if 
\begin{align*}
\so\left(\left\langle A_{1}A_{2}\right\rangle ,\dots,\left\langle A_{n-1}A_{n}\right\rangle ,\left\langle A_{n}A_{1}\right\rangle \right) & \le n-2,\\
\so\left(\left\langle A_{n}A_{n+1}\right\rangle ,\left\langle A_{n+1}A_{1}\right\rangle ,\left\langle A_{n}A_{1}\right\rangle \right) & \le1.
\end{align*}
Applying now Lemma \ref{lem:s0_s1_split} to these inequalities and
adding the condition of Lemma \ref{lem:joint-of-two} for the consistency
of $\left\langle A_{n}A_{1}\right\rangle $ with $\left\langle A_{n}\right\rangle $
and $\left\langle A_{1}\right\rangle $, we obtain the following system
\[
\begin{array}{l}
\so\left(\left\langle A_{1}A_{2}\right\rangle ,\dots,\left\langle A_{n-1}A_{n}\right\rangle \right)+\left\langle A_{n}A_{1}\right\rangle \le n-2,\\
\se\left(\left\langle A_{1}A_{2}\right\rangle ,\dots,\left\langle A_{n-1}A_{n}\right\rangle \right)-\left\langle A_{n}A_{1}\right\rangle \le n-2,\\
\so\left(\left\langle A_{n}A_{n+1}\right\rangle ,\left\langle A_{n+1}A_{1}\right\rangle \right)+\left\langle A_{n}A_{1}\right\rangle \le1,\\
\se\left(\left\langle A_{n}A_{n+1}\right\rangle ,\left\langle A_{n+1}A_{1}\right\rangle \right)-\left\langle A_{n}A_{1}\right\rangle \le1,\\
\left|\left\langle A_{n}\right\rangle +\left\langle A_{1}\right\rangle \right|-1\le\left\langle A_{n}A_{1}\right\rangle \le1-\left|\left\langle A_{n}\right\rangle -\left\langle A_{1}\right\rangle \right|.
\end{array}
\]
Statement (i) holds if and only if this system is satisfied, for some
real value of $\left\langle A_{n}A_{1}\right\rangle $. And it is
satisfied if and only if 
\[
\begin{array}{l}
\left\{ \begin{array}{c}
\se\left(\left\langle A_{1}A_{2}\right\rangle ,\dots,\left\langle A_{n-1}A_{n}\right\rangle \right)-n+2\\
\se\left(\left\langle A_{n}A_{n+1}\right\rangle ,\left\langle A_{n+1}A_{1}\right\rangle \right)-1\\
\left|\left\langle A_{n}\right\rangle +\left\langle A_{1}\right\rangle \right|-1
\end{array}\right.\\
\\
\le\left\{ \begin{array}{c}
n-2-\so\left(\left\langle A_{1}A_{2}\right\rangle ,\dots,\left\langle A_{n-1}A_{n}\right\rangle \right),\\
1-\so\left(\left\langle A_{n}A_{n+1}\right\rangle ,\left\langle A_{n+1}A_{1}\right\rangle \right),\\
1-\left|\left\langle A_{n}\right\rangle -\left\langle A_{1}\right\rangle \right|,
\end{array}\right.
\end{array}
\]
with the inequality holding for any left-hand expression combined
with any right-hand expression. The inequalities with matching rows
are satisfied always: the first two because 
\[
\se(a_{1}\dots,a_{n})+\so(a_{1},\dots,a_{n})=2\sum_{k=1}^{n}\left|a_{k}\right|-2\min_{k}\left|a_{k}\right|\le2n-2
\]
for $a_{1},\dots,a_{n}\in[-1,1]$; the third one due to the fact that
\[
\left|a+b\right|+\left|a-b\right|=\max\left(2\left|a\right|,2\left|b\right|\right)\leq2
\]
for $a,b\in\left[-1,1\right]$. This leaves the following six inequalities
\[
\begin{array}{l}
\se\left(\left\langle A_{1}A_{2}\right\rangle ,\dots,\left\langle A_{n-1}A_{n}\right\rangle \right)-n+2\\
\le1-\so\left(\left\langle A_{n}A_{n+1}\right\rangle ,\left\langle A_{n+1}A_{1}\right\rangle \right),
\end{array}
\]
\[
\se\left(\left\langle A_{1}A_{2}\right\rangle ,\dots,\left\langle A_{n-1}A_{n}\right\rangle \right)-n+2\le1-\left|\left\langle A_{n}\right\rangle -\left\langle A_{1}\right\rangle \right|,
\]
\[
\begin{array}{l}
\se\left(\left\langle A_{n}A_{n+1}\right\rangle ,\left\langle A_{n+1}A_{1}\right\rangle \right)-1\\
\le n-2-\so\left(\left\langle A_{1}A_{2}\right\rangle ,\dots,\left\langle A_{n-1}A_{n}\right\rangle \right),
\end{array}
\]
\[
\left|\left\langle A_{n}\right\rangle +\left\langle A_{1}\right\rangle \right|-1\le n-2-\so\left(\left\langle A_{1}A_{2}\right\rangle ,\dots,\left\langle A_{n-1}A_{n}\right\rangle \right),
\]
\[
\left|\left\langle A_{n}\right\rangle +\left\langle A_{1}\right\rangle \right|-1\le n-2-\so\left(\left\langle A_{1}A_{2}\right\rangle ,\dots,\left\langle A_{n-1}A_{n}\right\rangle \right),
\]
\[
\left|\left\langle A_{n}\right\rangle +\left\langle A_{1}\right\rangle \right|-1\le1-\so\left(\left\langle A_{n}A_{n+1}\right\rangle ,\left\langle A_{n+1}A_{1}\right\rangle \right).
\]
They simplify to 
\[
\begin{array}{l}
s_{0}\left(\left\langle A_{1}A_{2}\right\rangle ,\dots,\left\langle A_{n-1}A_{n}\right\rangle \right)\\
+s_{1}\left(\left\langle A_{n}A_{n+1}\right\rangle ,\left\langle A_{n+1}A_{1}\right\rangle \right)\le n-1,
\end{array}
\]
\[
s_{0}\left(\left\langle A_{1}A_{2}\right\rangle ,\dots,\left\langle A_{n-1}A_{n}\right\rangle \right)+\left|\left\langle A_{n}\right\rangle -\left\langle A_{1}\right\rangle \right|\le n-1,
\]
\[
\begin{array}{l}
s_{0}\left(\left\langle A_{n}A_{n+1}\right\rangle ,\left\langle A_{n+1}A_{1}\right\rangle \right)\\
+s_{1}\left(\left\langle A_{1}A_{2}\right\rangle ,\dots,\left\langle A_{n-1}A_{n}\right\rangle \right)\le n-1,
\end{array}
\]
\[
s_{0}\left(\left\langle A_{n}A_{n+1}\right\rangle ,\left\langle A_{n+1}A_{1}\right\rangle \right)+\left|\left\langle A_{n}\right\rangle -\left\langle A_{1}\right\rangle \right|\le2,
\]
\[
s_{1}\left(\left\langle A_{1}A_{2}\right\rangle ,\dots,\left\langle A_{n-1}A_{n}\right\rangle \right)+\left|\left\langle A_{n}\right\rangle +\left\langle A_{1}\right\rangle \right|\le n-1,
\]
\[
s_{1}\left(\left\langle A_{n}A_{n+1}\right\rangle ,\left\langle A_{n+1}A_{1}\right\rangle \right)+\left|\left\langle A_{n}\right\rangle +\left\langle A_{1}\right\rangle \right|\le2,
\]
and we combine pairs of inequalities using Lemma \ref{lem:s0_s1_split}
to obtain 
\begin{align}
s_{1}\left(\left\langle A_{1}A_{2}\right\rangle ,\dots,\left\langle A_{n-1}A_{n}\right\rangle ,\left\langle A_{n}A_{n+1}\right\rangle ,\left\langle A_{n+1}A_{1}\right\rangle \right) & \le n-1,\label{eq:n+1-cycle}\\
s_{1}\left(\left\langle A_{1}A_{2}\right\rangle ,\dots,\left\langle A_{n-1}A_{n}\right\rangle ,\left\langle A_{n}\right\rangle ,\left\langle A_{1}\right\rangle \right) & \le n-1,\label{eq:n-chain}\\
s_{1}\left(\left\langle A_{n}A_{n+1}\right\rangle ,\left\langle A_{n+1}A_{1}\right\rangle ,\left\langle A_{n}\right\rangle ,\left\langle A_{1}\right\rangle \right) & \le2.\label{eq:3-chain}
\end{align}
These three inequalities are satisfied if and only if Statement (i)
holds. In particular, Statement (i) implies (\ref{eq:n+1-cycle}),
and this completes the proof by induction of the necessity part of
the theorem: for any $n>1$, if $A_{1},\dots,A_{n}$ are jointly distributed
with expectations (\ref{eq:expectations_n}) then these expectations
satisfy (\ref{eq:main lemma}) (and Lemma \ref{lem:joint-of-two}).

Now, Corollary~\ref{cor:chain-joint} implies that a joint distribution
of $A_{1},\dots,A_{n}$ with expectations $\left\langle A{}_{1}\right\rangle ,\ldots,\left\langle A{}_{n}\right\rangle $,
$\left\langle A_{1}A_{2}\right\rangle ,\dots,\left\langle A_{n-1}A_{n}\right\rangle $
(satisfying Lemma \ref{lem:joint-of-two}) always exists. If we close
this chain into a cycle by introducing a constant variable $A_{n+1}\equiv1$,
we get $n+1$ jointly distributed variables with expectations $\left\langle A{}_{1}\right\rangle ,\ldots,\left\langle A{}_{n}\right\rangle ,\left\langle A{}_{n+1}\right\rangle =1$,
$\left\langle A_{1}A_{2}\right\rangle ,\dots,\left\langle A_{n-1}A_{n}\right\rangle ,\left\langle A_{n}A_{n+1}\right\rangle =\left\langle A_{n}\right\rangle ,\left\langle A_{n+1}A_{1}\right\rangle =\left\langle A_{1}\right\rangle $.
Applying to it the just established necessary part of the theorem,
we conclude that \eqref{eq:n-chain} always holds. Similarly, considering
the chain $A_{n},A_{n+1},A_{1}$ (whose joint distribution always
exists) and adding the constant variable $A'\equiv1$ to close the
chain into a cycle, the necessary condition implies \eqref{eq:3-chain}
with $A'\equiv1$. Thus, \eqref{eq:3-chain} also holds always, leaving
just \eqref{eq:n+1-cycle} as the equivalent condition for Statement
(i). 
\end{proof}

\begin{proof}[\textbf{{Proof of the main criterion (Theorem 4).}}]
From Theorem \ref{thm:A1-An}, contexts $\left\{ \left(R_{i}^{i},R_{i\oplus1}^{i}\right):i=1,\ldots,n\right\} $
and connections $\left\{ \left(R_{i}^{i},R_{i}^{i\ominus1}\right):i=1,\ldots,n\right\} $
with specified expectations $\left\{ \left\langle R_{i}^{i}R_{i\oplus1}^{i}\right\rangle :i=1,\ldots,n\right\} $
and $\left\{ \left\langle R_{i}^{i\ominus1}R_{i}^{i}\right\rangle :i=1,\ldots,n\right\} $
(subject to Lemma \ref{lem:joint-of-two}) can be imposed a joint
distribution upon if and only if 
\begin{equation}
\so\left(\left\langle R_{i}^{i}R_{i\oplus1}^{i}\right\rangle ,\left\langle R_{i}^{i\ominus1}R_{i}^{i}\right\rangle :i=1,\ldots,n\right)\le2n-2.
\end{equation}
As the variables of the connection $\left(R_{i}^{i},R_{i}^{i\ominus1}\right)$
are dichotomic, the probability of them being equal can be written
as $\Pr\left[R_{i}^{i}=R_{i}^{i\ominus1}\right]=\left(1+\left\langle R_{i}^{i\ominus1}R_{i}^{i}\right\rangle \right)/2$
and so this probability is maximized if and only if the expectation
$\left\langle R_{i}^{i\ominus1}R_{i}^{i}\right\rangle $ is maximized.
By Lemma \ref{lem:joint-of-two}, $1-\left|\left\langle R_{i}^{i}\right\rangle -\left\langle R_{i}^{i\ominus1}\right\rangle \right|$
is the maximum possible value of $\left\langle R_{i}^{i\ominus1}R_{i}^{i}\right\rangle $
given the distributions of $R_{i}^{i}$ and $R_{i}^{i\ominus1}$ (determined
by $\left\langle R_{i}^{i}\right\rangle $ and $\left\langle R_{i}^{i\ominus1}\right\rangle $).
The statement of the theorem now follows from Definition 2. \end{proof}
\begin{cor}[to Theorem~4]
\label{cor:-necessary} A cyclic system of rank $n>1$ with dichotomic
random variables has a maximally noncontextual description only if
\begin{equation}
\so\big(\left\langle R_{i}^{i}R_{i\oplus1}^{i}\right\rangle :i=1,\ldots,n\big)-\sum_{i=1}^{n}\left|\left\langle R_{i}^{i}\right\rangle -\left\langle R_{i}^{i\ominus1}\right\rangle \right|\le n-2.\label{eq:simpleform-2}
\end{equation}
\end{cor}
\begin{proof}
Using Lemma~\ref{lem:s0_s1_split}, 
\begin{align*}
2n-2 & \geq\so\left(\left\langle R_{i}^{i}R_{i\oplus1}^{i}\right\rangle ,1-\left|\left\langle R_{i}^{i}\right\rangle -\left\langle R_{i}^{i\ominus1}\right\rangle \right|:i=1,\ldots,n\right)\\
 & \geq\so\big(\left\langle R_{i}^{i}R_{i\oplus1}^{i}\right\rangle :i=1,\ldots,n\big)\\
 & \quad+\se\big(1-\left|\left\langle R_{i}^{i}\right\rangle -\left\langle R_{i}^{i\ominus1}\right\rangle \right|:i=1,\ldots,n\big),
\end{align*}
and 
\begin{align*}
 & \se\big(1-\left|\left\langle R_{i}^{i}\right\rangle -\left\langle R_{i}^{i\ominus1}\right\rangle \right|:i=1,\ldots,n\big)\\
 & \geq\sum_{i=1}^{n}\left(1-\left|\left\langle R_{i}^{i}\right\rangle -\left\langle R_{i}^{i\ominus1}\right\rangle \right|\right)\\
 & =n-\sum_{i=1}^{n}\left|\left\langle R_{i}^{i}\right\rangle -\left\langle R_{i}^{i\ominus1}\right\rangle \right|.
\end{align*}
\end{proof}
\begin{cor}[to Theorem 4.]
A cyclic consistently-connected system of rank $n>1$ with dichotomic
random variables has a maximally noncontextual description if and
only if 
\begin{equation}
\so\left(\left\langle R_{i}^{i}R_{i\oplus1}^{i}\right\rangle :i=1,\ldots,n\right)\le n-2.\label{eq:CC_appendix}
\end{equation}
\end{cor}
\begin{proof}
For consistently-connected systems, the main criterion has the form
\begin{equation}
2n-2\geq\so\Big(\left\langle R_{1}^{1}R_{2}^{1}\right\rangle ,\ldots,\left\langle R_{n}^{n}R_{1}^{n}\right\rangle ,\underbrace{1,\ldots,1}_{n\textnormal{ times}}\Big).
\end{equation}
The form \ref{eq:CC_appendix} follows from the easily verifiable
general formula 
\[
\so\left(x_{1},\ldots,x_{k}\right)=\sum_{i=1}^{k}\left|x_{i}\right|-2[x_{1}\cdots x_{k}>0]\min\left(\left|x_{1}\right|,\ldots,\left|x_{k}\right|\right),
\]
where {[}...{]} is the Iverson bracket, equal to 1 if the predicate
within it is true, and zero otherwise,\end{proof}
\begin{cor}[to Theorem 4]
A cyclic consistently-connected system of rank $n=5$ with dichotomic
random variables and with 
\[
\Pr\left[R_{i}^{i}=1,R_{i\oplus1}^{i}=1\right]=0,\quad i=1,\ldots,5,
\]
has a maximally noncontextual description if and only if the original
KCBS inequality holds, 
\begin{equation}
K=\sum_{i=1}^{5}p_{i}\leq2,\label{eq:Klyachko inequality}
\end{equation}
where $p_{i}=\Pr\left[R_{i}^{i}=1\right]=\Pr\left[R_{i}^{i\ominus1}=1\right]$,
$i=1,\ldots,5$. 
\end{cor}
The expression \eqref{eq:Klyachko inequality} for $K$ follows from
\eqref{eq:CC_appendix} and $\left\langle R_{i}^{i}R_{i\oplus1}^{i}\right\rangle =1-2\left(p_{i}+p_{i\oplus1}\right)$.
The proof that (\ref{eq:Klyachko inequality}) and (\ref{eq:CC_appendix})
are equivalent is obtained by considering linear combinations $L=\sum_{k=1}^{5}\iota_{k}\left[1-2\left(p_{k}+p_{k\oplus1}\right)\right]$
with 1, 3, and 5 negative $\iota_{k}$'s, in accordance with the definition
of $\so$, and showing that paired inequalities $L>3,K\leq2$ and
$L\leq3,K>2$ have no solutions satisfying $p_{i}\geq0$ and $p_{i}+p_{i\oplus1}\leq1$,
$i=1,\ldots,5$.


\begin{thebibliography}{10}
\bibitem{remark_conjecture}This formula is in fact equivalent to
\eqref{eq:master}, as conjectured in Ref.~\cite{DK_Bell_LG_K} and
proved in Ref.~\cite{conjecture}.
\end{thebibliography}
\end{document}